\documentclass[letterpaper,conference]{IEEEtran}

\usepackage{amsmath}
\usepackage{amssymb, color}
\usepackage{amsthm, graphicx}
\usepackage{algorithm,algorithmic}

\newtheorem{prop}{Proposition}
\newtheorem{lem}[prop]{Lemma}
\newtheorem{thm}[prop]{Theorem}

\theoremstyle{definition}
\newtheorem{defi}[prop]{Definition}

\newtheorem{ex}[prop]{Example}
\newtheorem{rem}[prop]{Remark}

\newcommand{\F}{\mathbb{F}}

\newcommand{\rk}{\mathrm{rank}}

\newcommand{\rs}{\mathrm{rs}}

\newcommand{\qdeg}{\mathrm{qdeg}}

\newcommand{\Lp}{\mathcal{L}_q(x,q^m)}

\IEEEoverridecommandlockouts

\title{List-Decoding Gabidulin Codes via Interpolation and the Euclidean Algorithm}
\author{\IEEEauthorblockN{Margreta Kuijper and Anna-Lena Trautmann\thanks{ALT is also with the Department of Electrical and Computer Systems Engineering, Monash University. She was supported by Swiss National Science Foundation Fellowship no. 147304.}} \IEEEauthorblockA{Department of Electrical and Electronic Engineering, University of Melbourne, Australia.}}

\begin{document}

\maketitle

\begin{abstract}
We show how Gabidulin codes can be list decoded by using a parametrization approach. For this we consider a certain module in the ring of linearized polynomials and find a minimal basis for this module using the Euclidean algorithm with respect to composition of polynomials. For a given received word, our decoding algorithm computes a list of all codewords that are closest to the received word with respect to the rank metric.
\end{abstract}

\section{Introduction}

Gabidulin codes are a family of optimal rank-metric codes, useful in different fields of coding theory, e.g.\ in (random) linear network coding \cite{si08j}, space-time coding \cite{lu03}, crisscoss error correction \cite{ro91} and distributed storage \cite{si12}. They were first derived by Gabidulin in \cite{ga85a} and independently by Delsarte in \cite{de78}. 
These codes can be seen as the $q$-analog of Reed-Solomon codes, using $q$-linearized polynomials instead of arbitrary polynomials over the finite field $\F_q$ (where $q$ is a prime power). They are optimal in the sense that they are not only MDS codes with respect to the Hamming metric, but also achieve the Singleton bound with respect to the rank metric and are thus MRD codes.

There has been a rising interest in the last decade due to their application in network coding \cite{ko08,si08j}. Since then a lot of work has been done on how to decode these codes. The question of minimum distance decoding inside the unique decoding radius has been addressed e.g.\ in  \cite{ga85a,ga92,lo06,ri04p,si10,si11,si09p}, whereas the more general setting of list decoding, beyond the unique decoding radius, is investigated in e.g.\ \cite{lo06p,ma12,wa13p,xi11}. Related work on list-decoding lifted Gabidulin codes can be found in \cite{tr13p}.

In this work we explore list decoding further and, in contrast to the Sudan-Guruswami approach of \cite{ma12,wa13p},  present a parametric approach analogous to the one for list decoding Reed-Solomon codes from \cite{al11}. In a similar way as \cite{lo06} we use interpolation, however unlike \cite{lo06} we perform list decoding rather than unique decoding. A difference between our paper and the papers \cite{lo06,xi11} is that our approach is based on the Euclidean algorithm. A more important difference with \cite{xi11} is that our decoding method yields all closest codewords, rather than just one. The latter is due to our parametrization approach.

The paper is structured as follows: In the following section we introduce $q$-linearized polynomials, Gabidulin codes, the rank metric and state some known properties of those. Moreover we explain the error span polynomial and recall the interpolation based unique decoding set-up for Gabidulin codes from \cite{lo06}. In Section \ref{sec:interpolation} we derive the module of $q$-linearized polynomials containing all those polynomials that interpolate the received word and show that finding all elements of this module fulfilling certain requirements is equivalent to list decoding with respect to the rank metric. In Section \ref{sec:algorithm} we describe a list decoding algorithm based on the previously described interpolation module using the Euclidean algorithm for $q$-linearized polynomials. We conclude this paper in Section \ref{sec:conclusion}.


\section{Preliminaries}

Let $q$ be a prime power and let $\F_q$ denote the finite field with $q$ elements. It is well-known that there always exists a primitive element $\alpha$ of the extension field $\F_{q^m}$, such that $\F_{q^m}\cong \F_q[\alpha] $. Moreover, $\F_{q^m}$  is isomorphic (as a vector space) to the vector space $\F_q^m$. If not noted differently we will use the isomorphism
\begin{align*}
\F_q^m &\longrightarrow \F_{q^m}\cong \F_q[\alpha]  \\
(v_1, \dots, v_m) &\longmapsto \sum_{i=1}^m v_i \alpha^{i-1}    .
\end{align*}
One then easily gets the isomorphic description of matrices over the base field $\F_q$ as vectors over the extension field, i.e.\ $\F_q^{m\times n}\cong \F_{q^m}^n$. Since we will work with matrices over different underlying fields we denote the rank of a matrix $X$ over $\F_q$ by $\rk_q(X)$.

For some vector $(v_1,\dots, v_n) \in \F_{q^m}^n$ we denote the $k \times n$ \emph{Moore matrix} by
\[M_k(v_1,\dots, v_n) := \left( \begin{array}{cccc}  v_1 & v_2 &\dots &v_n \\ v_1^{[1]} & v_2^{[1]} &\dots &v_n^{[1]} \\ \vdots \\  v_1^{[k-1]} & v_2^{[k-1]} &\dots &v_n^{[k-1]} \end{array}\right)   ,\]
where $[i]:= q^i$. A \emph{$q$-linearized polynomial} over $\F_{q^m}$ is defined to be of the form
\[f(x) = \sum_{i=0}^{n} a_i x^{[i]}   \quad, \quad a_i \in\F_{q^m} , \]
where $n$ is called the \emph{$q$-degree} of $f(x)$, assuming that $a_n\neq 0$, denoted by $\qdeg (f)$. This class of polynomials was first studied by Ore in \cite{or33}. 
One can easily check that $f(x_1 + x_2)= f(x_1)+f(x_2)$ and $f(\lambda x_1) = \lambda f(x_1)$ for any $x_1,x_2 \in \F_{q^m}$ and $\lambda \in \F_q$, hence the name \emph{linearized}. The set of all $q$-linearized polynomials over $\F_{q^m}$ is denoted by $\Lp$. This set forms a non-commutative ring with the normal addition $+$ and composition $\circ$ of polynomials. 
Because of the non-commutativity, products and quotients of elements of $\Lp$ have to be specified as being "left" or "right" products or quotients. To not be mistaken with the standard division, we call the inverse of the composition \emph{symbolic division}. I.e.\ $f(x)$ is symbolically divisible on the right by $g(x)$ with quotient $m(x)$ if $$ g(x) \circ m(x) = g(m(x)) = f(x).$$
Efficient algorithms for all these operations (left and right symbolic multiplication and division) exist and can be found e.g.\ in \cite{ko08}.

\begin{lem}[cf.\ \cite{li97b} Thm. 3.50]\label{lem:rootspace}
Let $f(x) \in \Lp$ and $\F_{q^s}$ be the smallest extension field of $\F_{q^m}$ that contains all roots of $f(x)$. Then the set of all roots of $f(x)$ forms a $\F_q$-linear vector space in $\F_{q^s}$.
\end{lem}

\begin{lem}[\cite{li97b} Thm. 3.52]\label{lem:nullpoly}
Let $U$ be a $\F_q$-linear subspace of $\F_{q^m}$. Then 
\( \prod_{\beta \in U} (x-\beta)\)
is an element of $\Lp$.
\end{lem}

Note that, if $\beta_1,\dots,\beta_t$ is a basis of $U$, one can rewrite 
$$ \prod_{\beta \in U} (x-\beta) = \lambda \det(M_{t+1}(\beta_1,\dots,\beta_t,x))$$ 
for some constant $\lambda\in\F_{q^m}$.

Let $g_1,\dots, g_n \in \F_{q^m}$ be linearly independent over $\F_q$. We define a \emph{Gabidulin code} $C\subseteq \F_{q^m}^{n}$ as the linear block code with generator matrix $M_k(g_1,\dots, g_n)$.        
Using the isomorphic matrix representation we can interpret $C$ as a matrix code in $\F_q^{m\times n}$.The \emph{rank distance} $d_R$ on  $\F_q^{m\times n}$ is defined by
\[d_R(X,Y):= \rk_q(X-Y) \quad, \quad X,Y \in \F_q^{m\times n} \]
and analogously for the isomorphic extension field representation. 
It holds that the code $C$ constructed before has dimension $k$ over $\F_{q^m}$ and minimum rank distance (over $\F_q$) $n-k+1$. One can easily see by the shape of the parity check and the generator matrices that an equivalent definition of the code is
\[C =  \{(f(g_1),\dots,f(g_n))\in \F_{q^m}^n \mid f(x) \in \Lp_{<k}  \} ,\]
where $\Lp_{<k} := \{f(x) \in \Lp, \qdeg(f(x)) < k\}$. 
For more information on bounds and constructions of rank-metric codes the interested reader is referred to \cite{ga85a}.

Consider a received word $\mathbf r = (r_1,\dots,r_n) \in \F_{q^m}^n$ as the sum $\mathbf r = \mathbf c + \mathbf e$, where $\mathbf c = (c_1,\dots,c_n)\in C$ is a codeword and $\mathbf e = (e_1,\dots,e_n)\in \F_{q^m}^n$ is the error vector. The following statement was formulated in a similar, but less general, manner in Theorem $1$ in \cite{lo06}. 
\begin{thm}\label{thm2}
Let $f(x)\in \Lp, \qdeg(f(x))< k$ and $c_i=f(g_i)$ for $i=1,\dots,n$.
It holds that $d_R(\mathbf c, \mathbf r) = t$ if and only if there exists a $D(x) \in \Lp$, such that $ \qdeg(D(x))= t$ and
\[D(r_i) = D(f(g_i)) \quad \forall i\in\{1,\dots,n\}.\]
Furthermore, this $D(x)$ is unique.
\end{thm}
\begin{proof}
Let $D(x) \in \Lp$ such that $D(r_i) = D(f(g_i))$ and $\qdeg(D(x)) =t$. This implies that $D(r_i - f(g_i)) = 0$ for all $i$. Define $e_i := r_i - f(g_i)$, then $e_i\in\F_{q^m}$ and every element of $\langle e_1,\dots, e_n\rangle$ is a root of $D(x)$ (see Lemma \ref{lem:rootspace}). Since $D(x)$ is non-zero and has degree $q^t$, it follows that the linear space of roots has $q$-dimension $t$, which implies that $(e_1,\dots,e_n)$ has rank  $t$. This means that the rank distance between $(c_1,\dots,c_n)$ and $(r_1,\dots,r_n)$ is equal to $t$. Thus, one direction is proven.

For the other direction let $(c_1,\dots,c_n), (r_1,\dots,r_n)$ have rank distance $t$, i.e. $(e_1,\dots,e_n) := (c_1-r_1,\dots, c_n-r_n)$ has rank $t$. Then by Lemma \ref{lem:nullpoly} there exists a non-zero $D(x)\in \Lp$ of degree $q^t$ such that $D(e_i)=0$ for all $i$. 
By linearity we get that $D(c_i)=D(r_i)$ for $i=1,\dots,n$. Since we know that $c_i=f(g_i)$, the statement follows.
The uniqueness of $D(x) = \prod_{\beta\in \langle e_1,\dots, e_n\rangle} (x-\beta)$ follows from the fact that its degree is equal to the number of its distinct roots.
\end{proof}

\begin{rem}
The previous theorem states that the roots of $D(x)$ form a vector space of degree $t$ which is equal to the span of $e_1,\dots,e_n$. This is why $D(x)$ is also called the \emph{error span polynomial} (cf.\ e.g.\ \cite{si09}). The analogy in the classical Hamming metric set-up is the \emph{error locator polynomial}, whose roots indicate the locations of the errors, and whose degree equals the number of errors.
\end{rem}

The interpolation-based unique decoding algorithm for Gabidulin codes from Loidreau \cite{lo06} can now be formulated as follows. 
Assume that $\rk_q(\mathbf e)=d_R(\mathbf c, \mathbf r) < d_R(C)/2$, i.e.\ that $\mathbf r$ is within the unique decoding radius. 
Find all pairs $(N(x), D(x)) \in \Lp^2$ with $\qdeg(N) < k+\rk_q(\mathbf e) \leq  (n+k)/2$ and $ \qdeg(D(x))\leq (n-k)/2$,  and check if $N(x)$ is symbolically divisible on the right by $D(x)$. If such a couple is found, then $D(x)$ is a valid error span polynomial by Theorem \ref{thm2}, and the symbolic quotient of $N(x)$ and $D(x)$ is the $q$-linearized polynomial that corresponds to the sent message.

In the next section we move beyond the unique decoding of \cite{lo06} and describe an interpolation-based decoding algorithm that is able to find all closest codewords, within or beyond the unique decoding radius.


\section{The Interpolation Module}\label{sec:interpolation}



For the rest of the paper 
let $g_1,\dots, g_n \in\F_{q^m}$ be linearly independent over $\F_q$ and let $M_k(g_1,\dots,g_n)$ be the generator matrix of the Gabidulin code $C\subseteq \F_{q^m}^n$. Let  $\mathbf{r}=(r_1,\dots,r_n) \in \F_{q^m}^n$ be the received word.


For our following investigations we need a $q$-linearized analog of the Lagrange polynomial. 

\begin{defi}
Let $\mathbf g=(g_1,\dots,g_n)$ and define the matrix $\mathfrak{D}_i(\mathbf g, x) :=  M_{n}(g_1,\dots,g_n,x)$ without the $i$-th column.
 We define the \emph{$q$-Lagrange polynomial} as  
\[\Lambda_{\mathbf g, \mathbf r}(x) := \sum_{i=1}^n (-1)^{n-i}  r_i \frac{\det(\mathfrak{D}_i(\mathbf g, x))}{\det (M_n(\mathbf g))} \quad \in \F_{q^m}[x] .\]
\end{defi}

\begin{lem}\label{lem:Lagrange}
Consider the setting of the previous definition. 
Then $\Lambda_{\mathbf g, \mathbf{r}}(x)\in \Lp$, i.e.\ it is $q$-linearized. Moreover, $\Lambda_{\mathbf g, \mathbf{r}}(g_i) = r_i $ for $i=1,\dots,n$ and $\qdeg(\Lambda_{\mathbf g,\mathbf{r}}(x))=n-1$.
\end{lem}
\begin{proof}
Since $\det(\mathfrak{D}_i(\mathbf g, x))$ is $q$-linearized and $\Lambda_{\mathbf g, \mathbf{r}}(x)$ is the sum of scalar multiples of these determinants, $\Lambda_{\mathbf g, \mathbf{r}}(x)$ is also $q$-linearized. 
One can easily check that $\det(\mathfrak{D}_i(\mathbf g , x)) = (-1)^{n-i} \det (M_n(\mathbf g)) $ for $x=g_i$ and $\det(\mathfrak{D}_i(\mathbf g , x)) =0$ for $x= g_j$ where $j\neq i$. Hence for $x=g_i$ all but the $i$-th summand are zero and the $i$-th summand is equal to $r_i$.
\end{proof}

Furthermore we need the following fact.

\begin{lem}\label{lem3}
Let  $L(x) \in \Lp$, such that $L(g_i)=0$ for all $i$. Then
\[\exists H(x)\in \Lp : L(x) = H(x)\circ \prod_{\alpha \in \langle g_1,\dots,g_n\rangle}(x-\alpha)  . \]
\end{lem}
\begin{proof}
We know from Lemma \ref{lem:nullpoly} that $\prod_{\alpha \in \langle g_1,\dots,g_n\rangle}(x-\alpha) \in \Lp$. Moreover there always exists unique left and right division in $\Lp$, i.e.\ in this case there exist unique polynomials $H(x),R(x)\in \Lp$ such that $L(x) = H(x)\circ\prod_{\alpha \in \langle g_1,\dots,g_n\rangle}(x-\alpha) + R(x)$ and $\qdeg(R(x))< \qdeg \prod_{\alpha \in \langle g_1,\dots,g_n\rangle}(x-\alpha) =n$. Since any $\alpha \in \langle g_1,\dots,g_n\rangle$ is a root of $L(x)$ and of $\prod_{\alpha \in \langle g_1,\dots,g_n\rangle}(x-\alpha)$, they must also be a root of $R(x)$. Hence we have $q^n$ distinct roots for $R(x)$ and $\deg(R)<q^n$, thus $R(x) \equiv 0$ and the statement follows.
\end{proof}

In the following we abbreviate the row span of a (polynomial) matrix $A$ by $\rs(A)$.

\begin{defi}
Define the polynomials $\Pi(x):= \prod_{\alpha \in \langle g_1,\dots,g_n\rangle} (x-\alpha)$ and $\Lambda_{\bf g,r}(x)$ as the $q$-Lagrange polynomial, such that $\Lambda_{\bf g,r}(g_i)=r_i$ for all $i$. Furthermore define the left submodule of $\Lp$
\[\mathfrak{M}(\bf r) := \rs \left[\begin{array}{cc}  \Pi(x) & 0 \\ -\Lambda_{\bf g,r}(x) & x \end{array}\right].\]
We call $\mathfrak{M}(\bf r)$ the \emph{interpolation module} for $\mathbf r$. 
\end{defi}

\begin{defi}
We define the \emph{$(k_1,k_2)$-weighted $q$-degree} of  $[f(x) \;\; g(x)] \in  \mathfrak{M}(\bf r) $ as  $\max\{k_1+ \qdeg(f) , k_2+ \qdeg(g)\}$.
\end{defi}

We identify any $[f(x) \quad g(x)] \in  \mathfrak{M}(\bf r)$ with the bivariate linearized $q$-polynomial $Q(x,y)= f(x) + g(y)$.
We will now show that the name interpolation module is justified for $\mathfrak{M}(\bf r)$.

\begin{thm}\label{thm5}
$\mathfrak{M}(\bf r)$ consists exactly of all $Q(x,y)= f(x) +g(y)$ with $f(x), g(x) \in \Lp$, such that $Q(g_i,r_i)=0$ for $i=1,\dots,n$.
\end{thm}

\begin{proof}
For the first direction let $Q(x,y)= f(x) +g(y)$ be an element of $\mathfrak{M}(\bf r)$. Then there exist $a(x), b(x) \in \Lp$ such that $f(x) = a(x)\circ\Pi(x) - b(x)\circ\Lambda_{\bf g,r}(x)$ and $b(x) = g(x)$, thus $Q(g_i,r_i)= a(\Pi(g_i)) - b(\Lambda_{\bf g,r}(g_i)) + b(r_i) = 0 - b(r_i) + b(r_i) = 0$.

For the other direction let $f(x), g(x) \in \Lp$ be such that $Q(g_i,r_i)= f(g_i) +g(r_i)=0$ for $i=1,\dots,n$. To show that $Q(x,y)\in\mathfrak{M}(\bf r)$ we need to find $a(x) \in \Lp$ such that
\[a(x)\circ \Pi(x) - b(x)\circ \Lambda_{\bf g,r}(x) = f(x) \quad\textnormal{ and }\quad  b(x) = g(x) .\]
We substitute the second into the first equation to get 
\begin{align}
a(x)\circ \Pi(x)  = f(x)+ g(x)\circ \Lambda_{\bf g,r}(x) .
\end{align}
By assumption it holds that $f(g_i)+g( \Lambda_{\bf g,r}(g_i) )= f(g_i)+g(r_i) = 0$ for all $i$. Then, by Lemma \ref{lem3}, it follows that $f(x)+ g(x)\circ \Lambda_{\bf g,r}(x)$ is symbolically divisible on the right by $\Pi(x)$ and hence there exists $a(x)\in \Lp$ such that $(1)$ holds.
\end{proof}

Combining all the previous results we get a description of all codewords with distance $t$ to the received word in the new parametrization:

\begin{thm}\label{thm:main}
The elements $[N(x) \quad -D(x)]$ of $\mathfrak{M}(\bf r)$  that fulfill
\begin{enumerate}
\item $\qdeg(N(x))\leq t+k-1$,
\item $\qdeg(D(x))=t$,
\item $N(x)$ is symbolically divisible on the right by $D(x)$, i.e.\ there exists $f(x)\in\Lp$ such that $D(f(x))=N(x)$,
\end{enumerate}
are in one-to-one correspondence with the codewords of rank distance $t$ to $\mathbf r$. 
\end{thm}
\begin{proof}
Let $\mathbf c \in \F_{q^m}^n$ be a codeword such that $d_R(\mathbf c, \mathbf r)=t$ with the corresponding message polynomial $f(x)\in\Lp_{<k}$. Then by Theorem \ref{thm2} there exists $D(x)\in\Lp$ of $q$-degree $t$ such that $D(f(g_i))=D(r_i)$ for $i=1,\dots,n$. By Theorem \ref{thm5} we know that $[D(f(x)) \quad -D(x)]$ is in $ \mathfrak{M}(\bf r)$. It holds that $\qdeg(D(f(x)))\leq t+k-1$ and that $(D(f(x))$ is divisible on the right by $D(x)$.

On the other hand let $[N(x) \quad -D(x)] \in \mathfrak{M}(\bf r)$ fulfil conditions $1)-3)$. Then we know that the divisor $f(x)\in \Lp$ has $q$-degree less than $k$ and it holds $N(x)=D(f(x))$. Since it is in $ \mathfrak{M}(\bf r)$ we know by Theorem \ref{thm5} that $D(f(g_i))-D(r_i)=0$ for all $i$ and hence by Theorem \ref{thm2} that $d_R(\mathbf c, \mathbf r)=t$, if $\bf c$ is the codeword corresponding to the message polynomial $f(x)$.
\end{proof}

\begin{rem}
The two first conditions in the previous theorem  imply that the $(0,k-1)$-weighted $q$-degree of $[N(x) \quad -D(x)]$ is equal to $t+k-1$.
\end{rem}

Therefore, we have shown in this section that list decoding within rank radius $t$ is equivalent to finding all elements  $[N(x) \quad -D(x)]$ in $\mathfrak{M}(\bf r)$ with $(0,k-1)$-weighted $q$-degree less than or equal to $t+k-1$ and $\qdeg(N(x))\leq \qdeg(D(x))+k-1$, such that $N(x)$ is symbolically divisible on the right by $D(x)$.
It follows that, to find all closest codewords to a given $\mathbf r \in \F_{q^m}^n$, we need to find all elements  $[N(x) \quad -D(x)] \in \mathfrak{M}(\bf r)$ of minimal $(0,k-1)$-weighted $q$-degree such that $\qdeg(N(x))\leq \qdeg(D(x))+k-1$ and $N(x)$ is symbolically divisible on the right by $D(x)$.


\section{The Algorithm}\label{sec:algorithm}

We can now describe the list decoding algorithm. Since in most applications you want to find the set of all closest codewords to the received word, our algorithm will do exactly this. In contrast, a complete list decoder with a prescribed radius $t$ finds all codewords within radius $t$ from the received word, even if some of them are closer than others.

We recall that our approach is analogous to \cite{al11}, where a minimal Gr\"obner basis approach is taken. In fact, for linearized polynomials this minimal Gr\"obner basis approach can be formulated in exactly the same way, replacing multiplication by composition and redefine `degree' by `$q$-degree'.
Due to space limitations we omit the details. Whenever we mention `minimal basis' in the sequel, we mean `minimal Gr\"obner basis' in this generalized sense.

Algorithm \ref{alg1} describes the decoding algorithm. It will iteratively search for all elements in $\mathfrak{M}(\bf r)$ of $(0,k-1)$-weighted $q$-degree $t+k-1$ for increasing $t$ and check the requirements of Theorem \ref{thm:main}. As soon as solutions are found, $t$ will not be increased and the algorithm terminates.

We first present our decoding algorithm under the assumption that we can find a minimal basis for the interpolation module. We then detail the construction of such a basis in Algorithm \ref{alg2}. 
Note that we use the notation $g(x) = [g^{(1)}(x) \;\; g^{(2)}(x)]$ for elements of the interpolation module  $\mathfrak{M}(\bf r)$.

\begin{algorithm}
\caption{Minimal list decoding of Gabidulin codes.}
\label{alg1}
\begin{algorithmic}
\REQUIRE Received word $\bf r \in \F_{q^m}^n$.
\STATE 1. Compute $\Pi (x)$ and $\Lambda_{\bf g,r}(x)$, both in $\Lp$. Define the interpolation module $$\mathfrak{M}(\bf r) := \rs \left[\begin{array}{cc}  \Pi(x) & 0 \\ -\Lambda_{\bf g,r}(x) & x \end{array}\right]  .$$
\STATE 2. Compute a minimal basis $G=\{g_1(x), g_2(x)\}$ of $\mathfrak{M}(\bf r) $ with respect to the $(0,k-1)$-weighted degree, with $\qdeg(g_2^{(1)}(x))\leq \qdeg(g_2^{(2)}(x))+k-1$. 
\STATE 3. Define $\ell_1, \ell_2$ as the  $(0,k-1)$-weighted degrees of $g_1(x), g_2(x)$, respectively.
\STATE 4. Define \textbf{list}$:=[]$ (an empty list) and $j:=0$.
\WHILE{\textbf{list}$=[]$}
\FORALL{$a(x)\in \Lp, \qdeg(a(x))\leq \ell_2-\ell_1+j$}
\FORALL{monic $b(x) \in \Lp, \qdeg(b(x))=j$}
\STATE $f(x) := a(x)\circ g_1(x) + b(x)\circ g_2(x)$
\IF{$f^{(1)}(x)$ is symb. (right) divisible by $f^{(2)}(x)$} 
\STATE add the respective symb. quotient to \textbf{list}
\ENDIF
\ENDFOR
\ENDFOR
\STATE $j:=j+1$
\ENDWHILE
\RETURN \textbf{list}
%
%
\end{algorithmic}
\end{algorithm}

\begin{thm}
Algorithm \ref{alg1} yields a list of all message polynomials such that the corresponding codeword is closest to the received word.
\end{thm}
\begin{proof}
Let $t$ be such that $d_R(\mathbf c,\mathbf r)=t$ for a closest codeword $\mathbf c$.
Note that the variable $j$ in the algorithm corresponds to $t-\ell_2 + k-1$. If we substitute this for $j$, then we get that 
\begin{small}
\begin{align*}
\qdeg(f^{(1)}(x))\leq \hspace{6.5cm}\\
\max\{\qdeg(a(x))+\qdeg(g_1^{(1)}(x)), \qdeg(b(x)) + \qdeg(g_2^{(1)}(x))\}\\
\leq \ell_2+j = t+k-1
\end{align*}
\end{small}
and, since $\qdeg(g_2^{(1)}(x))\leq \qdeg(g_2^{(2)}(x))+k-1$ implies that  $\ell_2-k+1=\qdeg(g_2^{(2)}(x))$,
\begin{small}
\begin{align*}
\qdeg(f^{(2)}(x))=\hspace{6.5cm}\\
\max\{\qdeg(a(x))+\qdeg(g_1^{(2)}(x)), \qdeg(b(x)) + \qdeg(g_2^{(2)}(x))\} \\
=\ell_2+j-k+1 =t  .
\end{align*}
\end{small}
Hence, $f^{(1)}(x)$ fulfills requirement $1)$ and $f^{(2)}(x)$ requirement $2)$  in Theorem \ref{thm:main}. In fact, it can be proven that $G$ is a minimal Gr\"obner basis for the interpolation module and has the so-called \emph{Predictable Leading Monomial Property} analogous to \cite{al11,ku11}. As a result of this property, the elements in the two \textbf{for}-loops that fulfill the divisibility requirement correspond to codewords with rank distance $t=j-\ell_2+k-1$ from $\bf r$. 
Due to space limitations we refrain from proving this in detail.

Moreover, increasing $j$ by one is equivalent to increasing $t$ by one. Therefore, once we have solutions in the list, the algorithm terminates, since elements added to the list at stage $j+1$ would be further away then the ones added at stage $j$.

It remains to show that there are no codewords at rank distance less than $k-1-\ell_2$, since this is the distance for the initial loops with $j=0$. 
Assume there would be such a codeword with corresponding message polynomial $m(x)\in\Lp$. Then there exists $D(x)\in\Lp$ with $q$-degree less than $k-1-\ell_2$ such that $g'(x) := [D(m(x)) \;\; D(x)]$ is in $\mathfrak{M}(\bf r)$. Then the $(0,k-1)$-weighted $q$-degree of $g'(x)$ is less than $\ell_2$, which means that $G$ is not a minimal basis of $\mathfrak{M}(\bf r)$, which is a contradiction.
\end{proof}


\begin{thm}
Algorithm \ref{alg2} below produces a minimal Gr\"obner basis for our interpolation module $\mathfrak{M}(\bf r)$ via the Euclidean algorithm for $q$-linearized polynomials, replacing multiplication by composition.
\end{thm}
 For the sake of brevity we omit the proof of this result.


\begin{algorithm}[hb]
\caption{Computation of $g_1, g_2$ via the (linearized) Euclidean Algorithm.}
\label{alg2}
\begin{algorithmic}
\REQUIRE Received word $\bf r$; polynomials $\Pi(x)$ and $\Lambda_{\bf g,r}(x)$.
\STATE Initialize $j=0$ and defined the linearized polynomials $h_0(x),h_1(x),t_0(x), t_1(x)$ as 
\[ \left[\begin{array}{cc}h_0(x) & t_0(x) \\ h_1(x) & t_1(x) \end{array}\right]  :=  \left[\begin{array}{cc}  \Pi(x) & 0 \\ -\Lambda_{\bf g,r}(x) & x \end{array}\right]  . \]
	\WHILE{$\qdeg(t_{j+1}) + k-1 < \qdeg(h_{j+1})$}
\STATE Apply the (linearized) Euclidean algorithm to compute the linearized polynomials $q_{j+1}(x)$ and $h_{j+2}(x)$ such that $h_j(x) = q_{j+1}(h_{j+1}(x)) + h_{j+2}(x)$ and $\qdeg(h_{j+2})< \qdeg(h_{j+1})$.
\STATE Update $t_{j+2}(x):= t_j (x) - q_{j+1}(t_{j+1}(x))$.
\STATE Set $j:= j+1$.
\ENDWHILE
\RETURN $g_1 := [\: h_j(x) \quad t_j(x)\:]$ and $g_2 :=  [\: h_{j+1}(x) \quad t_{j+1}(x)\:]$
\end{algorithmic}
\end{algorithm}

\begin{ex}\label{ex15}
Consider the Gabidulin code in $\F_{2^3}\cong \F_2[\alpha]$ (with $\alpha^3=\alpha +1$) with generator matrix
\[G= \left( \begin{array}{ccc} 1 & \alpha & \alpha^2 \\ 1 & \alpha^2 & \alpha^4\end{array}\right) \]
and the received word
\[\mathbf{ r} =(\: \alpha +1 \;0 \; \alpha \:) .\]
Then we construct the interpolation module
$$M(\bf r) = \rs \left[\begin{array}{cc}  \Pi(x) & 0 \\ -\Lambda_{\bf r}(x) & x \end{array}\right]  = \rs \left[\begin{array}{cc}  x^8 + x & 0 \\  \alpha^2 x^4+ \alpha^5 x & x \end{array}\right] .$$
To compute a minimal basis we use the Euclidean algorithm and get
\[x^8+x  = (\alpha^3 x^2) \circ ( \alpha^2 x^4+ \alpha^5 x) + \alpha^6 x^2 +  x .\]
Since $\qdeg( \alpha^3 x^2) +k-1 =2 \geq 1=\qdeg( \alpha^6 x^2 +  x)$, the algorithm terminates and a minimal basis (w.r.t. the $(0,1)$-weighted $2$-degree) of this module is
\[ \left[\begin{array}{cc}g_1^{(1)} & g_1^{(2)}\\ g_2^{(1)} &g_2^{(2)}\end{array}\right]  =  \left[\begin{array}{cc} \alpha^2 x^4+ \alpha^5 x & x\\ \alpha^6 x^2 +  x  & \alpha^3 x^2 \end{array}\right] . \]
Hence we get $\ell_1 =2$ and $ \ell_2= 2$, i.e.\ we want to use all $a(x)\in\mathcal{L}_2(x,2^3)$ with $2$-degree less than or equal to $0$ and all monic $b(x)\in\mathcal{L}_2(x,2^3)$  with $2$-degree equal to $0$. Thus, $a(x)= a_0 x$ for $a_0\in\F_{2^3}$ and $b(x)= x$. We get divisibility for $a_0\in\F_{2^3}\backslash \{0\}$. The corresponding message polynomials and codewords are
\[m_1(x) =  x^2 + \alpha x \quad, \quad c_1=(\: \alpha^3 \; 1 \; \alpha^3),\]
\[m_2(x) = \alpha^5 x^2 + \alpha^2 x \quad, \quad c_1=(\: \alpha^3  \; \alpha \; \alpha),\]
\[m_3(x) = \alpha^3 x^2 + \alpha^4 x \quad, \quad c_1=(\: \alpha^2 +1 \; 0 \; \alpha^2),\]
\[m_4(x) = \alpha^4 x^2  \quad,\quad c_3=(\: \alpha^2 +\alpha \; \alpha^2+1 \; \alpha),\]
\[m_5(x) = \alpha^6 x^2 + \alpha^6 x \quad, \quad c_1=(\: 0 \; \alpha^3 \; 1),\]
\[m_6(x) = \alpha x^2 +  x \quad,\quad c_2=(\: \alpha^3 \; 1 \; \alpha^3).\]
All these codewords are rank distance $1$ away from $\bf r$.
\end{ex}

Note that in the previous example all output codewords are only rank distance $1$ away from $\bf r$, but the Hamming distance between them and $\bf r$ can vary between $1,2$ or even $3$.

\section{Conclusion}\label{sec:conclusion}

In this paper we introduced a novel interpolation based decoding algorithm for Gabidulin codes with respect to the rank metric. For this we construct the interpolation module for a given received word and find a minimal basis of this module with respect to the $(0,k-1)$-weighted $q$-degree, utilizing the Euclidean algorithm for composition of linearized polynomials. Then we check the divisibility requirement for certain combinations of the two basis elements to get the list of all closest codewords to that received word. 
To our knowledge the Euclidean algorithm has not been used before to do this type of list decoding for rank-metric Gabidulin decoding.

Future work consists of a detailed complexity analysis; it is anticipated that the method is efficient particularly when the decoding radius is close to the unique decoding radius, such as in one-step ahead decoding cases, illustrated by Example \ref{ex15}.



\bibliographystyle{plain}
\bibliography{margreta_anna-lena}

\end{document}